\documentclass[12pt]{iopart}
\pdfoutput=1
\usepackage[utf8]{inputenc}
\usepackage[english]{babel}
\usepackage[T1]{fontenc}
\usepackage{iopams}
\usepackage{setstack}
\usepackage{amsthm}
\usepackage{makeidx}
\usepackage{graphicx}
\usepackage{color}
\usepackage{tikz-cd}
\usepackage{hyperref}
\hypersetup{citecolor=blue}
\hypersetup{colorlinks=true}
\hypersetup{linkcolor=blue}
\hypersetup{urlcolor=blue}
\usepackage[numbers,sort&compress]{natbib}

\newcommand{\eqref}[1]{(\ref{#1})}

\newcommand{\RR}{\mathbb{R}}
\newcommand{\tmin}{\dot{\otimes}}

\newcommand{\tmax}{\hat{\otimes}}
\newcommand{\id}{\mathrm{id}}
\newcommand{\linspan}{\mathrm{span}}
\newcommand{\aff}{\mathrm{aff}}
\newcommand{\conv}{\mathrm{conv}}
\newcommand{\meas}{\mathcal{M}}
\newcommand{\Ha}{\mathcal{H}}
\newcommand{\dens}{\mathcal{D}}
\newcommand{\effect}{\mathcal{E}}
\newcommand{\bound}{\mathcal{B}}
\usepackage{dsfont}
\newcommand{\I}{\mathds{1}}

\newcommand*{\bra}[1]{\langle #1 |}
\makeatletter
\newcommand*{\ket}[1]{| #1 \rangle \@ifnextchar\bra{\!}{}}
\makeatother

\newcommand*{\ketbra}[1]{| #1 \rangle \! \langle #1 |}
\newtheorem{theorem}{Theorem}
\newtheorem{proposition}{Proposition}
\newtheorem{lemma}{Lemma}
\newtheorem{corollary}{Corollary}
\theoremstyle{remark}
\newtheorem{example}{Example}
\theoremstyle{definition}
\newtheorem{definition}{Definition}

\begin{document}

\title[Incompatibility in restricted operational theories]{Incompatibility in restricted operational theories: connecting contextuality and steering}

\author{Martin Pl\'{a}vala}
\address{Naturwissenschaftlich-Technische  Fakult\"{a}t, Universit\"{a}t Siegen, 57068 Siegen, Germany}
\ead{martin.plavala@uni-siegen.de}

\begin{abstract}
We investigate the connection between steering and contextuality in general probabilistic theories. We show that for a class of bipartite states the steerability of the state by given set of measurements is equivalent to non-existence of preparation noncontextual hidden variable model for certain restricted theory constructed from the given state and measurements. The connection between steering and contextuality is provided by the concept of incompatibility in restricted theories, which we also investigate.
\end{abstract}


\section{Introduction}
One potential solution to the lack of intuitive interpretations of quantum theory is to assume that quantum theory is not fundamental, but that there is a hidden classical theory which can explain any prediction made by quantum theory. Such hidden classical theory is often called hidden variable model. Existence of hidden variable models was investigated in the past: it was shown by John Bell that local hidden variable models are incompatible with predictions of quantum theory \cite{Bell-ineq} and it was shown by Kochen and Specker \cite{KochenSpecker-hiddenVariables} that noncontextual hidden variable models are also incompatible with the predictions of quantum theory. Here noncontextual means that the probability of observing a measurement outcome in the hidden variable theory does not depend on other possible measurement outcomes. For a recent review of noncontextuality in quantum theory see \cite{BudroniCabelloGuhneKleinmann-contextuality}.

One possible operational version of contextuality was recently defined and discussed in \cite{SchmidSelbyWolfeKunjwalSpekkens-noncontextuality}. By operational we mean that instead of investigating the existence of noncontextual description of quantum theory, we investigate the existence of noncontextual description of some convex operational theory. The class of convex operational theories we have in mind are called general probabilistic theories, or GPTs for short \cite{Muller-review,Plavala-review}; they are large class of operational theories including both classical and quantum theories, but also theories that contain PR boxes \cite{PopescuRohrlich-PRbox,Barrett-GPTinformation} or postquantum steering \cite{CavalcantiSelbySikoraGalleySainz-Witworld}. Noncontextuality was investigated in GPTs before \cite{Spekkens-contextuality,Spekkens-toyTheory,ChiribellaYuan-contextuality,SchmidSpekkensWolfe-nonconIneq,SchmidSelbyPuseySpekkens-nonconModels,SchmidSelbySpekkens-causalInferentialTheories,GittonWoods-unitSeparability}. It is clear that some GPTs must have noncontextual hidden variable model; trivial example is any classical theory which is itself its own hidden variable model. Using the language of GPTs we can ask which theories apart from classical do allow given form of noncontextual hidden variable model. This is an important question as if it turns out that almost no theories allow for noncontextual hidden variable model, then one can interpret the result of Kochen and Specker as a mathematical statement about a class of non-classical theories, rather than as a result about quantum theory itself.

Incompatibility is one of the fundamental non-classical features of quantum theory; we say that two measurements are incompatible if they cannot be measured jointly. Incompatibility was heavily studied within quantum theory \cite{HeinosaariMiyaderaZiman-compatibility, GuhneHaapasaloKraftPellonpaaUola-incompatibility} but also in GPTs \cite{FilippovHeinosaariLeppajarvi-compatibility,Jencova-incomaptibility,BluhmJencovaNechita-incomaptibility}; it is for example known that incompatibility is necessary for steering \cite{UolaMoroderGuhne-steering,QuintinoVertesiBrunner-steering} and for violations of Bell inequalities \cite{BrunnerCavalcantiPironioScaraniWehner-BellNonlocality}.

In this paper we formulate and investigate different forms of noncontextuality and their connection to the notion of incompatibility in restricted GPTs, we show that there is an inherent connection between incompatibility and contextuality in all GPTs. Our results on compatibility generalize the results obtained in \cite{TavakoliUola-contextuality} and they are complementary to \cite{SelbySchmidWolfeSainzKunjwalSpekkens-contextuality} where it was shown that incompatibility is not necessary for the investigated type of contextuality. Our main result is that we apply the aforementioned connection to steering, where we show that steerability of certain class of states by given set of measurements directly corresponds to non-existence of preparation noncontextual hidden variable model for restricted theory constructed from the given state and measurements.

The paper is organized as follows: in Sec.~\ref{sec:GPT} we review the formalism of GPTs and in Sec.~\ref{sec:HVM} we introduce the three different types of noncontextual hidden variable models in GPTs. In Sec.~\ref{sec:incompat} we introduce two different types of incompatibility in restricted GPTs and we prove that there is a connection between incompatibility and non-existence of preparation noncontextual hidden variable. Finally in Sec.~\ref{sec:steering} we put all of the concepts together and we prove the connection between steerability of a state by given measurement and non-existence of preparation noncontextual hidden variable model.

\section{Restricted GPTs} \label{sec:GPT}
In this section we will introduce restricted general probabilistic theories (GPTs). This section is intended mainly to set the notation and introduce only the necessary concepts; for an exhaustive review of GPTs see \cite{Muller-review,Plavala-review}.

A state space $K$ is a compact convex subset of a finite-dimensional real vector space $V$. Given a state space $K$ we define the effect algebra $E(K)$ as the set of affine functions $f:K \to [0,1]$. An important element of $E(K)$ is the constant function $1_K \in E(K)$ given as $1_K(\rho) = 1$ for all $\rho \in K$. The interpretation is that $E(K)$ corresponds to all mathematically well-defined two-outcome measurements on $K$, i.e., for $f \in E(K)$ and $\rho \in K$ the number $f(\rho)$ is the probability of obtaining the `yes' outcome of the corresponding two-outcome measurement, the probability of `no' outcome is given by the normalization as $1-f(\rho)$.

Under reasonable assumptions one can show that we can use the tensor product to describe states shared between two parties. For a given state spaces $K_A$ and $K_B$ let $K_{AB}$ be the bipartite state space that contains all possible states that the two parties can share, then
\begin{equation} \label{eq:GPT-tensorSpan}
K_{AB} \subset \linspan(\{ \rho_A \otimes \rho_B: \rho_A \in K_A, \rho_B \in K_B \}).
\end{equation}
It is natural to assume that $K_{AB}$ contains all separable states, i.e., that we have $\rho_A \otimes \rho_B \in K_{AB}$ for all $\rho_A \in K_A$ and $\rho_B \in K_B$. This leads to the definition of the minimal tensor product
\begin{equation}
K_A \tmin K_B = \conv(\{\rho_A \otimes \rho_B : \rho_A \in K_A, \rho_B \in K_B\})
\end{equation}
and we then require $K_A \tmin K_B \subset K_{AB}$. We will also assume that $K_{AB}$ allows for all separable measurements, i.e., that for all $f_A \in E(K_A)$ and $f_B \in E(K_B)$ we have $(f_A \otimes f_B)(\rho_{AB}) \geq 0$. This leads to the definition  of the maximal tensor product
\begin{eqnarray}
K_A \tmax K_B = \{ &\psi \in \linspan(K_A \tmin K_B): (1_{K_A} \otimes 1_{K_B})(\psi) = 1, \nonumber \\
&(f_A \otimes f_B)(\psi) \geq 0, \forall f_A \in E(K_A), \forall f_B \in E(K_B) \}
\end{eqnarray}
and we then require $K_{AB} \subset K_A \tmax K_B$. Using linearity we can define the experiment where Alice measures but Bob doesn't: let $\rho_{AB} \in K_{AB}$, then using \eqref{eq:GPT-tensorSpan} we have $\rho_{AB} = \sum_{i=1}^n \alpha_i \sigma^A_i \otimes \sigma^B_i$ where $\alpha_i \in \RR$, $\sigma^A_i \in K_A$, $\sigma^B_i \in K_B$ for all $i \in \{1, \ldots, n\}$. Let $f_A \in E(K_A)$, we then define $(f_A \otimes \id_B)(\rho_{AB}) = \sum_{i=1}^n \alpha_i f_A(\sigma^A_i) \sigma^B_i$. Here $\id_B$ is the identity map on Bob's system.

\begin{example}[Classical theory]
A useful example of a state space is the classical theory where the state space is a simplex $S_n$. Here a simplex is a convex hull of linearly independent pure states, that is, $S_n = \conv(\{s_i\}_{i=1}^n)$ where $s_i$ are linearly independent vectors. The effect algebra $E(S_n)$ is then generated by the dual basis $b_j \in E(S_n)$, which is given by $b_j(s_i) = \delta_{ij}$ where $\delta_{ij}$ is the Kronecker symbol. We then have $1_{S_n} = \sum_{j=1}^n b_j$. Note that we will also use the following notation: let $I$ be any index set, e.g., $I = \{1, \ldots, n\}$, then we will use $S_I$ to denote the simplex generated by the states indexed by $I$, i.e., $S_I = \conv(\{s_i\}_{i \in I})$.
\end{example}

\begin{example}[Quantum theory]
Quantum theory is another example of a GPT. Here the state space is the set of all density matrices, i.e, we have $K = \dens(\Ha)$, where $\dens(\Ha) = \{ \rho \in \bound(\Ha): \rho \geq 0, \Tr(\rho) = 1 \}$, $\bound(\Ha)$ is the vector space of hermitian operators on the finite-dimensional complex Hilbert space $\Ha$ and $A \geq 0$ means that $A$ is positive semidefinite. We also have $E(\dens(\Ha)) = \effect(\Ha) = \{ M \in \bound(\Ha): 0 \leq M \leq \I \}$ up to an isomorphism; the value of the effect $M \in \effect(\Ha)$ on the state $\rho \in \dens(\Ha)$ is given by the Born rule $\Tr(\rho M)$. Let $K_A = \dens(\Ha_A)$ and $K_B = \dens(\Ha_B)$, then the bipartite state space is given by the tensor product of the underlying Hilbert spaces, that is $K_{AB} = \dens(\Ha_A \otimes \Ha_B)$.
\end{example}

Let $K_A$ and $K_B$ be state spaces, then a channel is an affine map $\Phi: K_A \to K_B$. A measurement is a channel where the target state space is a simplex, i.e., measurement is a channel $m: K_A \to S_n$. One can prove that every measurement $m$ corresponds to a tuple of effects $f_i \in E(K)$, $i \in \{1, \ldots, n\}$, such that $\sum_{i=1}^n f_i = 1_K$ and for $\rho \in K_A$ we have $m(\rho) = \sum_{i=1}^n f_i(\rho) s_i$. This is exactly the same as saying that $p(i)$, the probability of observing the outcome $i$, is given as $p(i) = f_i(\rho)$. Finally, a classical post-processing is a channel $\nu: S_n \to S_k$ that maps a simplex to another simplex. Every classical post-processing corresponds to a right stochastic matrix $\nu_{ij}$ , i.e., $\nu_{ij} \geq 0$ and $\sum_{j=1}^k \nu_{ij} = 1$, the correspondence is given as $\nu(s_i) = \sum_{j=1}^n \nu_{ij} s_j$ where $s_i$ is an extreme point of $S_n$.

Given a channel $\Phi: K_A \to K_B$ we can define the adjoint map $\Phi^*: E(K_B) \to E(K_A)$ as the unique channel such that $(\Phi^*(f))(\rho) = f(\Phi(\rho))$ holds for all $f \in E(K_B)$ and $\rho \in K_A$.

A restricted GPT is a theory in which not all mathematically well-defined states and measurements are allowed. Restricted GPTs arise in experimental scenarios, where we can usually prepare only a certain subset of all possible states and perform only a certain subset of all possible measurements. In order to describe the restricted theory we will use the pair $(K,E)$ where $K$ is a state space containing all preparable states and $E$ is the effect algebra containing all allowed effects, i.e., all allowed two-outcome measurements. Note that the pair $(K,E)$ does not completely describe the restricted theory, because we can also have independent higher-order restrictions on measurements with more than two outcomes \cite{FilippovGudderHeinosaariLeppajarvi-restrictions}. In order to also take into account the restrictions on the set of more than two-outcome measurements, we will denote $\meas(E)$ the set of all measurements allowed in the given restricted theory $(K,E)$. Going the other way, we can also define $E$ as the smallest effect algebra containing effects corresponding to measurements $m_x$, where $x$ is a some index; in this case we denote $E = \mathrm{effect}(m_x)$. If the set $\{ m_x \}$ is closed with respect to post-processing of measurements we get $\meas(\mathrm{effect}(m_x)) = \{ m_x \}$.

For practical reason we will consider a theory with no restrictions as a special case of a general restricted theory. Moreover we will abuse the notation and we will use $(K, E(K))$ to denote the theory with no restrictions.

Let $(K,E)$ be a restricted theory, then in order to make all probabilities positive we must require that $E \subset E(K)$, i.e., that every $f \in E$ is an affine function $f:K \to [0,1]$. Following the consistency results derived in \cite{FilippovGudderHeinosaariLeppajarvi-restrictions} we will also assume that $E$ is convex and $0_K, 1_K \in E$, where $0_K(\rho) = 0$ for all $\rho \in K$. One also has to consider whether $E$ contains enough measurements to distinguish any two states in $K$. Thus we arrive to the following definition:
\begin{definition}
Let $(K,E)$ be a restricted theory. We say that $(K,E)$ is tomographically complete if for every $\rho, \sigma \in K$, $\rho \neq \sigma$, there is $f \in E$ such that $f(\rho) \neq f(\sigma)$.
\end{definition}
It is reasonable to expect that in most practical applications the restricted theories $(K,E)$ will be tomographically complete. We will not assume topographical completeness because it is not needed for the validity of our results.

Given the restricted theory $(K,E)$ it is useful to define the state space of $E$ as the set of all linear functionals $\psi$ on $\linspan(E)$ such that $\psi$ is positive on $E$ and normalized, $\psi(1_K) = 1$. We get
\begin{equation}
S(E) = \{ \psi \in \linspan(E)^*: \psi(f) \geq 0, \forall f \in E, \psi(1_K) = 1 \}
\end{equation}
where $V^*$ is the dual of the vector space $V$. For a theory with no restrictions $(K, E(K))$ we have $S(E(K)) = K$ up to an isomorphism \cite{Plavala-review}. If $(K,E)$ is tomographically complete theory, then $K \subset S(E)$ up to an isomorphism. Finally we prove the following lemma which is going to be useful for working with tomographically incomplete restricted theories.
\begin{lemma} \label{lemma:GPTs-channelInRestricted}
Let $(K,E)$ be a restricted theory, then there is a channel $\Phi_{(K,E)}: K \to S(E)$ such that $m(\rho) = m(\Phi_{(K,E)}(\rho))$ for all $\rho \in K$ and for all measurements $m \in \meas(E)$.
\end{lemma}
\begin{proof}
Let $f \in E$ and $\rho \in K$. By construction of the restricted theory $(K,E)$ there must be a way to evaluate $f$ on $\rho$, i.e., the expression $f(\rho)$ must be well-defined. We will now employ a standard trick: rather then seeing $f$ as the function and $\rho$ as the input we will exchange their roles and we will interpret $\rho$ as the function and $f$ as the input. This is done using the evaluation map $\xi: \rho \to \xi_\rho$, where $\xi_\rho: E \to \RR$ is a function defined as $\xi_\rho(f) = f(\rho)$. It is straightforward to prove that $\xi_\rho: E \to \RR$ is linear and so $\xi_\rho \in \linspan(E)^*$. Similarly we get $\xi_\rho(f) \geq 0$ for all $f \in E$, $\xi_\rho(1_K) = 1$ and that the evaluation map $\xi$ is affine. It thus follows that $\xi_\rho \in S(E)$ and $\xi = \Phi_{(K,E)}$ is the channel we were looking for.
\end{proof}
The interpretation of Lemma \ref{lemma:GPTs-channelInRestricted} is that we can formally interpret every measure-and-prepare experiment in the restricted theory $(K,E)$ as preparation of some state $\rho \in K$, applying the fixed channel $\Phi_{(K,E)}$ and measuring the final state $\Phi_{(K,E)}(\rho) \in S(E)$ using a measurement $m \in \meas(E)$. Note that while $\rho \in K$ and $m \in \meas(E)$ are given by the choice of the experiment, the channel $\Phi_{(K,E)}$ is always fixed by the structure of the restricted theory $(K,E)$. If $(K,E)$ is tomographically complete, then $K \subset S(E)$ and one usually gets $\Phi_{(K,E)} = \id$.

\section{Hidden variable models} \label{sec:HVM}
Given a restricted theory $(K,E)$ one can ask whether hidden variable models with certain properties exist; the properties in question are for example locality, superdeterminism \cite{DonadiHossenfelder-superdeterminism,HossenfelderPalmer-superdeterminism}, or noncontextuality. The hidden variable model is a classical theory that can model any preparation-measurement experiment performable in $(K,E)$. We will be interested in the following three different versions of noncontextual hidden variable models \cite{Spekkens-contextuality}: preparation noncontextual hidden variable models, measurement noncontextual hidden variable models and preparation and measurement noncontextual hidden variable models.

Consider the following scenario: an experimenter is given a description of the input data, according to the description the experimenter prepares the input state, measures and obtains statistics of the experiment. Let $S_I$ be the simplex that describes the classical input and $S_O$ the simplex that describes the classical output. The scenario in question can be described using the channels $P: S_I \to K$ and $m: K \to S_O$, where $P$ is the preparation map and $m$ is the measurement. Thus the whole experiment can be pictorially represented as
\begin{equation}
\begin{tikzcd}
S_I \arrow[d, "P"]     \\
K \arrow[d, "m"] \\
S_O                   
\end{tikzcd}
\end{equation}
The upside of including the input data and $S_I$ in out model is that the whole formulation is closer to being time-symmetric \cite{Hardy-timeSymmetry}. Note that $S_I$ contains the full information about the input of the state preparation which allows us to distinguish scenarios that are later indistinguishable. For example, assume that $K = \dens(\Ha)$, $\dim(\Ha) = 2$, is a two-level quantum system and let $S_I = \conv(\{ s_0, s_1, s_+, s_- \})$. Let $\ket{\pm} = \frac{1}{\sqrt{2}}(\ket{0}\pm\ket{1})$ and define the preparation map $P: S_I \to K$ as $P(s_0) = \ketbra{0}$, $P(s_1) = \ketbra{1}$, $P(s_+) = \ketbra{+}$, $P(s_-) = \ketbra{-}$. It then follows that $P \left( \frac{1}{2} s_0 + \frac{1}{2} s_1 \right) = P \left( \frac{1}{2} s_+ + \frac{1}{2} s_- \right)$ and so $\frac{1}{2} s_0 + \frac{1}{2} s_1$ and $\frac{1}{2} s_+ + \frac{1}{2} s_-$ lead to preparing indistinguishable states. But despite that $\frac{1}{2} s_0 + \frac{1}{2} s_1$ and $\frac{1}{2} s_+ + \frac{1}{2} s_-$ are perfectly distinguishable states of $S_I$. Hence at the level of $S_I$ one can distinguish preparation inputs that are undistinguishable at the level of $K$.

A hidden variable model for the scenario described by $P$ and $m$ is given by a simplex $S_\Lambda$, preparation $\eta_P: S_I \to S_\Lambda$ and measurement $\mu_m: S_\Lambda \to S_O$ such that the statistics of both experiments are the same, i.e., for all $s_i \in S_I$ we have $(m \circ P)(s_i) = (\mu_m \circ \eta_P)(s_i)$, which can be written in terms of probability distributions as
\begin{equation}
p(a|P,m, i) = \sum_{\lambda \in \Lambda} p(a|m, \lambda) p(\lambda|P, i)
\end{equation}
where $p(a|P,m, i)$ is the probability of obtaining an outcome $a$ of measurement $m$ when the input data is $i$ and preparation $P$ is used, $p(\lambda|P, i)$ is the probability of obtaining the value of the hidden variable $\lambda$ when modeling preparation $P$ with input data $i$ and $p(a|m, \lambda)$ is the probability of obtaining the outcome $a$ when the value of the hidden variable is $\lambda$ and when modeling the measurement $m$. Here $\eta_P$ and $\mu_m$ are essentially the maps $\eta_P: i \to p(\lambda|P,i)$ and $\mu_m: \lambda \to p(a|m,\lambda)$. Note that the preparation $\eta_P$ can depend on $P$ and the measurement $\mu_m$ cen depend on $m$, one can capture these dependencies by defining the maps $\eta: P \mapsto \eta_P$ and $\mu: m \mapsto \mu_m$. Pictorially we get
\begin{equation} \label{eq:HVM-freeHVM}
\begin{tikzcd}
S_I \arrow["P"']{d}[name=P]{} \arrow["\eta_P", bend left]{rd}[name=etaP]{}
\arrow["\eta",shorten >=10pt,Rightarrow,to path={(P) -- node[label=below:$\eta$] {} (etaP)}]{}
& \\
K \arrow["m"']{d}[name=m]{}
& S_\Lambda \arrow["\mu_m", bend left]{ld}[name=mum]{}
\arrow["\eta", shorten >=10pt,Rightarrow,to path={(m) -- node[label=above:$\mu$] {} (mum)}]{} \\
S_O
&
\end{tikzcd}
\end{equation}
where the double arrows denote supermaps, that is maps that map maps. The meaning of \eqref{eq:HVM-freeHVM} is that both paths from $S_I$ to $S_O$ lead to the same result. It is straightforward to see that the hidden variable model corresponding to \eqref{eq:HVM-freeHVM} always exists, for example one may construct such hidden variable model by calculating the outcomes of the experiment by hand, the hidden variable $\Lambda$ then corresponds to all possible (finite) calculations one can write down. This is not surprising as so far we did not require that the hidden variable model satisfies any properties except that $i$ and $a$ do not depend on $\lambda$ and it was already observed in \cite{KochenSpecker-hiddenVariables} that the hidden variable model should be required to have some additional properties in order to not be trivial.

One can require the hidden variable model to be preparation noncontextual, meaning that for all preparations $P: S_I \to K$ and all input data $s_i \in S_I$ the prepared distribution of hidden variables depends only on $P(s_i)$. This means that for any $s, s' \in S_I$ and $P,P': S_I \to K$ such that $P(s) = P'(s')$ we must have $\eta_P(s) = \eta_{P'}(s')$. It follows that one can define a map $\iota: K \to S_\Lambda$ as follows: let $P$ be a preparation $P:S_I \to K$, let $s \in S_I$ and denote $P(s_i) = \rho$, then we define $\iota(\rho) = \eta_P(s)$. This is well-defined because if there is a different preparation $P':S_I \to K$ and input data $s' \in S_I$ such that $P'(s') = \rho$, then $P'(s') = \rho = P(s)$ and so $\iota(\rho) = \eta_P(s) = \eta_{P'}(s')$. Moreover $\iota$ is affine map: let $\rho, \sigma \in K$ and let $P$ be a preparation such that $P(s) = \rho$ and $P(s') = \sigma$, we can always construct such preparation as a preprocessing of the preparations of $\rho$ and $\sigma$. Let $\lambda \in [0,1]$, we then have $P(\lambda s + (1-\lambda) s') = \lambda \rho + (1-\lambda) \sigma$ and
\begin{eqnarray}
\iota(\lambda \rho + (1-\lambda) \sigma) &= \eta_P(\lambda s + (1-\lambda) s') = \lambda \eta_P(s) + (1-\lambda) \eta_P(s') \nonumber \\
&= \lambda \iota(\rho) + (1-\lambda) \iota(\sigma).
\end{eqnarray}
Thus $\iota$ is a channel and we get the following definition:
\begin{definition}[Preparation noncontextuality]
A restricted theory $(K,E)$ is preparation noncontextual if there is a map $\iota: K \to S_\Lambda$ such that for any experiment given by preparation $P:S_I \to K$ and measurement $m:K \to S_O$ we have
\begin{equation} \label{eq:HVM-NCprepHVM-eq}
p(a|P,m,i) = \sum_{\lambda \in \Lambda} p(a|m, \lambda) b_\lambda( (\iota \circ P) (s_i)),
\end{equation}
where $b_\lambda \in E(S_\Lambda)$ is such that $b_\lambda(s_{\lambda'}) = \delta_{\lambda \lambda'}$. One can also denote $b_\lambda( (\iota \circ P) (s_i)) = p(\lambda|P(s_i))$, explicitly showing that the distribution of hidden variable $\lambda$ depends only on $P(s_i)$. \eqref{eq:HVM-NCprepHVM-eq} is the same as
\begin{equation} \label{eq:HVM-NCprepHVM-diag}
\begin{tikzcd}
S_I \arrow[d, "P"']
& \\
K \arrow["m"']{d}[name=m]{} \arrow[r, "\iota"]
& S_\Lambda \arrow["\mu_m", bend left]{ld}[name=mum]{}
\arrow["\eta", shorten >=10pt,Rightarrow,to path={(m) -- node[label=above:$\mu$] {} (mum)}]{} \\
S_O
&
\end{tikzcd}
\end{equation}
\end{definition}

Analogically one can require that the hidden variable model is measurement noncontextual, meaning that for all measurements $m: K \to S_O$ and all outcome labels $s_a \in S_O$ the distribution of the outcome data depends only on $f_a = m^*(b_a)$ and $\lambda$. This means that for any $b, b' \in E(S_O)$ and $m,m': K \to S_O$ such that $m^*(b) = (m')^*(b')$ we must have $\mu_m^*(b) = \mu_{m'}^*(b')$. It follows that there must exist a map $\kappa: E \to E(S_\Lambda)$ such that if $f = m^*(b)$, then $\kappa(f) = \mu_m(b)$, moreover one can show that $\kappa$ is affine in the same way as for $\iota$. We thus get the definition:
\begin{definition}[Measurement noncontextuality]
A restricted theory $(K,E)$ is measurement noncontextual if there is a map $\kappa^*: S_\Lambda \to S(E)$ such that for any experiment given by preparation $P:S_I \to K$ and measurement $m:K \to S_O$ we have
\begin{equation} \label{eq:HVM-NCmeasHVM-eq}
p(a|P,m,i) = \sum_{\lambda \in \Lambda} b_a((m \circ \kappa^*)(s_\lambda)) p(\lambda| P, i),
\end{equation}
where one can denote $b_a((m \circ \kappa^*)(s_\lambda)) = p(a|m^*(b_a))$. Diagramatic equivalent of \eqref{eq:HVM-NCmeasHVM-eq} is
\begin{equation} \label{eq:HVM-NCmeasHVM-diag}
\begin{tikzcd}
S_I \arrow["P"']{d}[name=P]{} \arrow["\eta_P", bend left]{rd}[name=etaP]{} \arrow["\eta",shorten >=10pt,Rightarrow,to path={(P) -- node[label=below:$\eta$] {} (etaP)}]{}
&
& \\
K \arrow["m"']{d}[name=m]{}
& S_\Lambda \arrow[r, "\kappa^*"]
& S(E) \arrow[lld, "m"] \\
S_O
&
&
\end{tikzcd}
\end{equation}
\end{definition}
The diagram \eqref{eq:HVM-NCmeasHVM-diag} is quite different from \eqref{eq:HVM-NCprepHVM-diag}, but one can make the diagrams similar by explicitly including the channel $\Phi_{(K,E)}$ given by Lemma \ref{lemma:GPTs-channelInRestricted} and treating $m$ as only measurement on $S(E)$.

One can also require both at the same time, i.e., one can require the hidden variable model to be both preparation noncontextual and measurement noncontextual. In this case we will say that $(K,E)$ is simplex-embeddable, as already defined in \cite{SchmidSelbyWolfeKunjwalSpekkens-noncontextuality}.
\begin{definition}[Simplex-embeddability]
A restricted theory $(K,E)$ is simplex-embeddable if there is a simplex $S_\Lambda$ and channels $\iota: K \to S_\Lambda$ and $\kappa^*: S_\Lambda \to S(E)$ such that for all $\rho \in K$ and $f \in E$ we have $(\kappa(f))(\iota(\rho)) = f(\rho)$, which is the same as
\begin{equation} \label{eq:HVM-SE-diag}
\begin{tikzcd}
S_I \arrow[d, "P"']
&
& \\
K \arrow[r, "\iota"] \arrow[d, "m"']
& S_\Lambda \arrow[r, "\kappa^*"]
& S(E) \arrow[lld, "m"] \\
S_O
&
&
\end{tikzcd}
\end{equation}
\end{definition}

\section{Two definitions of incompatibility in restricted operational theories} \label{sec:incompat}
Before defining incompatibility of measurements in restricted theories, lets remind ourselves of the definition used in theories with no restrictions. Although we will define compatibility of only two measurements, it is straightforward to extend the notion to finite or infinite amount of measurements.  Let $(K, E(K))$ be a theory satisfying no-restriction hypothesis and let $m_1$ and $m_2$ be two measurements corresponding to the sets of effects $f_i \in E(K)$ and $g_j \in E(K)$ respectively, where $i \in \{1, \ldots, n_1\}$, $j \in \{1, \ldots, n_2\}$ and $\sum_{i=1}^{n_1} f_i = \sum_{j=1}^{n_2} g_j = 1_K$. Then we say that $m_1$ and $m_2$ are compatible if there is a measurement $m$ given by the set of effects $h_{ij} \in E(K)$, $i \in \{1, \ldots, n_1\}$, $j \in \{1, \ldots, n_2\}$ such that
\begin{eqnarray}
f_i &= \sum_{j=1}^{n_2} h_{ij}, \label{eq:incompat-noresSum-fi} \\
g_j &= \sum_{i=1}^{n_1} h_{ij}. \label{eq:incompat-noresSum-gj}
\end{eqnarray}
Measurement $m$ satisfying \eqref{eq:incompat-noresSum-fi} and \eqref{eq:incompat-noresSum-gj} is called the joint measurement. The definition of compatibility is equivalent to requiring that there is a measurement $m$ given by the set of $h_\lambda$, $\lambda \in \Lambda$, where $\Lambda$ is some arbitrary index set, and conditional probability distributions $p(i|1, \lambda)$ and $p(j|2, \lambda)$ such that $f_i = \sum_{\lambda \in \Lambda} p(i|1, \lambda) h_{\lambda}$, $g_j = \sum_{\lambda \in \Lambda} p(j|2, \lambda) h_{\lambda}$, see \cite{FilippovHeinosaariLeppajarvi-compatibility} for a proof.

Given a restricted theory $(K,E)$ we have two options when defining compatibility of measurements: we either require that $h_{ij} \in E$ or we relax this to $h_{ij} \in E(K)$. One may be inclined to argue that $h_{ij} \in E$ is the correct option, as the joint measurement $m$ given by the effects $h_{ij}$ should be an allowed measurement, but, as we will see, the other option is also useful and quite interesting, since it implies that even though the joint measurement is not an allowed measurement, the available states from $K$ cannot witness \cite{CarmeliHeinosaariToigo-incWitness} the incompatibility of $m_1$ and $m_2$. Thus we obtain the following two definitions:
\begin{definition}[$E$-compatibility] \label{def:incompat-Ecompat}
Let $(K,E)$ be a restricted theory and let $m_1$ and $m_2$ be measurements given by $f_i \in E$ and $g_j \in E$ respectively, where $i \in \{1, \ldots, n_1\}$ and $j \in \{1, \ldots, n_2\}$. We say that $m_1$ and $m_2$ are $E$-compatible if there are $h_{ij} \in E$, $i \in \{1, \ldots, n_1\}$, $j \in \{1, \ldots, n_2\}$, such that \eqref{eq:incompat-noresSum-fi} and \eqref{eq:incompat-noresSum-gj} hold. The measurement $m$ will be called $E$-joint measurement of $m_1$ and $m_2$.
\end{definition}

\begin{definition}[$E(K)$-compatibility] \label{def:incompat-EKcompat}
Let $(K,E)$ be a restricted theory and let $m_1$ and $m_2$ be measurements given by $f_i \in E$ and $g_j \in E$ respectively, where $i \in \{1, \ldots, n_1\}$ and $j \in \{1, \ldots, n_2\}$. We say that $m_1$ and $m_2$ are $E(K)$-compatible if there are $h_{ij} \in E(K)$, $i \in \{1, \ldots, n_1\}$, $j \in \{1, \ldots, n_2\}$, such that \eqref{eq:incompat-noresSum-fi} and \eqref{eq:incompat-noresSum-gj} hold. The measurement $m$ will be called $E(K)$-joint measurement of $m_1$ and $m_2$.
\end{definition}

The following result is immediate:
\begin{lemma}
Let $(K,E)$ be a restricted theory and let $m_1$ and $m_2$ be measurements. If $m_1$ and $m_2$ are $E$-compatible, then they are also $E(K)$-compatible.
\end{lemma}
\begin{proof}
The proof follows from $E \subset E(K)$.
\end{proof}

\begin{figure}
\centering
\includegraphics[width=0.4\textwidth]{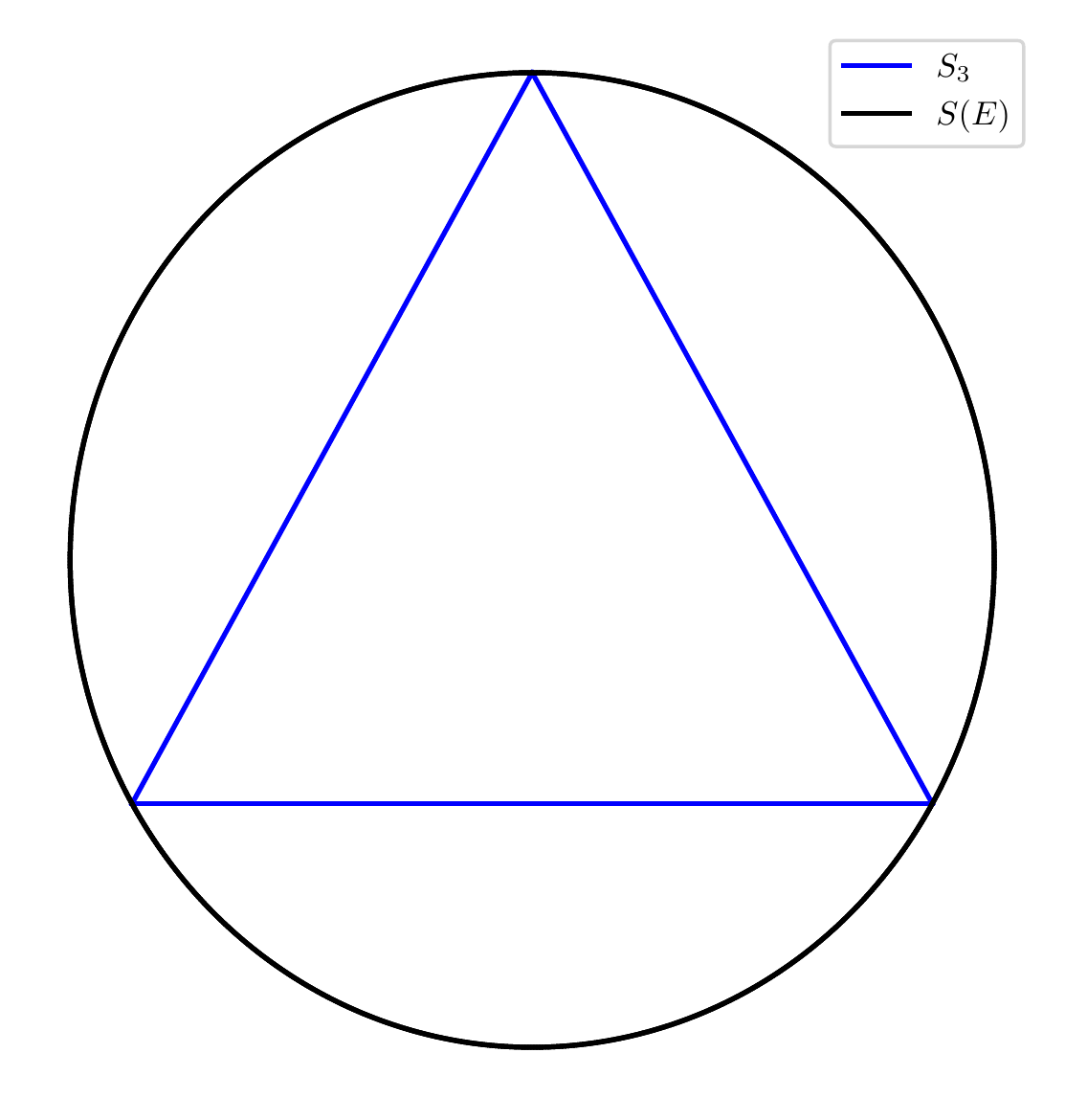}
\caption{Embedding of the classical state space $S_3$ inside the circle $S(E)$ used in Example~\ref{exm:incompat-triangleInCircle}.}
\label{fig:incompat-triangleInCircle}
\end{figure}

It is straightforward to construct a restricted theory $(K,E)$ and measurements $m_1$ and $m_2$ that are $E(K)$-compatible, but $E$-incompatible.
\begin{example}[Triangle-in-circle] \label{exm:incompat-triangleInCircle}
Let $K = S_3$ be a triangle and let $S(E)$ be a circumscribed circle of $K$, see Figure~\ref{fig:incompat-triangleInCircle}. Then $S(E)$ coincides with the state space of the real quantum theory and so there clearly are $E$-incompatible measurements $m_1$ and $m_2$. But since $K = S_3$ is a simplex, it follows that all measurements are $E(K)$-compatible.
\end{example}

In the following we will investigate what happens if all measurements in a restricted theory $(K,E)$ are either $E$-compatible or $E(K)$-compatible and we will compare these concepts to simplex-embeddability. We will take this path because, as we will see later, exactly these concepts become crucial when investigating the steerability of a bipartite state. Also note that we are not loosing too much generality by taking this path as when investigating the compatibility of measurement $m_1, \ldots, m_N$ one can always take $E$ to be the smallest effect algebra including all effects corresponding to the measurements $m_1, \ldots, m_N$. The first result is an immediate application of known results on theories with no-restrictions.
\begin{proposition} \label{prop:incompat-Ecompat-SEsimplex}
Let $(K,E)$ be a restricted theory. All measurements in $(K,E)$ are $E$-compatible if and only if $S(E)$ is a simplex.
\end{proposition}
\begin{proof}
If $S(E)$ is a simplex, then all measurements containing effects from $E$ are compatible and so all measurements in $(K,E)$ are $E$-compatible.

If all measurements are $E$-compatible, then it follows that all pairs of two-outcome measurements are $E$-compatible. Let $m_1$ and $m_2$ be a pair of two-outcome measurements, then their $E$-joint measurement $m$ is a well-defined measurement on $S(E)$. Then $(S(E),E)$ is a theory with no restrictions in which all pairs of two-outcome measurements are compatible. It follows that $S(E)$ is a simplex \cite{Plavala-simplex, Kuramochi-simplex}.
\end{proof}

The following result is similar to \cite[Theorem 1]{TavakoliUola-contextuality}.
\begin{proposition} \label{prop:incompat-EKcompat-NCprepHVM}
Let $(K,E)$ be a restricted theory. All measurements in $(K,E)$ are $E(K)$-compatible if and only if there exists a preparation noncontextual hidden variable model for $(K,E)$.
\end{proposition}
\begin{proof}
Let $\{ m_x \} = \meas(E)$, i.e., we index all measurements in $\meas(E)$ by the index $x$, and denote $f_{a|x}$ the effects corresponding to $m_x$. Since $m_x$ are $E(K)$-compatible, it follows that there are $h_\lambda \in E(K)$ such that $f_{a|x} = \sum_{\lambda \in \Lambda} p(a|x,\lambda) h_\lambda$. Let us define $\iota: K \to S_\Lambda$ by $\iota(\rho) = \sum_{\lambda \in \Lambda} h_\lambda(\rho) s_\lambda$ and $\mu_{m_x}:S_\Lambda \to S_O$ by $b_a(\mu_{m_x}(s_\lambda)) = p(a|x, \lambda)$. We then have
\begin{eqnarray}
p(a|P,m_x,i) &= f_{a|x}(P(s_i)) = \sum_{\lambda \in \Lambda} p(a|x,\lambda) h_\lambda(P(s_i)) \\
&= b_a((\mu_{m_x} \circ \iota)(P(s_i))).
\end{eqnarray}
Thus $\iota$ and $\mu: m_x \mapsto \mu_{m_x}$ constitute the preparation noncontextual hidden variable model.

Now assume that there is a preparation noncontextual hidden variable model for $(K,E)$. We then have that for every measurement $m:K \to S_O$ there are maps $\iota: K \to S_\Lambda$ and $\mu_m:S_\Lambda \to S_O$ such that $m = \mu_m \circ \iota$. Observe that $\iota$ is a measurement on $K$ corresponding to effects $h_\lambda \in E(K)$ and $\mu_m$ is just a classical post-processing, so any measurement $m$ in $(K,E)$ can be obtained as a classical post-processing $\mu_m$ of the measurement $\iota$. It follows that all measurements in $(K,E)$ are $E(K)$-compatible \cite[Subsection 5.2]{HeinosaariMiyaderaZiman-compatibility}.
\end{proof}

Finally, we get to describe the hierarchy between $E$-compatibility, simplex-embeddability and $E(K)$-compatibility.
\begin{theorem} \label{thm:incompat-hierarchy}
Let $(K,E)$ be a restricted theory. Then the following two implications hold:
\begin{enumerate}
\item\label{item:incompat-hierarchy-EtoSE} If all measurements on $(K,E)$ are $E$-compatible, then $(K,E)$ is simplex-embeddable.
\item\label{item:incompat-hierarchy-SEtoEK} If $(K,E)$ is simplex-embeddable, then all measurements on $(K,E)$ are $E(K)$-compatible.
\end{enumerate}
\end{theorem}
\begin{proof}
If all measurements in $(K,E)$ are $E$-compatible then $S(E)$ is a simplex, see Proposition~\ref{prop:incompat-Ecompat-SEsimplex}. The main idea used to prove \eqref{item:incompat-hierarchy-EtoSE} is to take $S_\Lambda = S(E)$, $\kappa^* = \id$ and $\iota = \Phi_{(K,E)}$, where $\Phi_{(K,E)}$ is the channel given by Lemma~\ref{lemma:GPTs-channelInRestricted}. Let $\rho \in K$ and $f \in E$, we then have $\kappa(f) = f$ and $f( \Phi_{(K,E)}(\rho) ) = f(\rho)$ follows from Lemma~\ref{lemma:GPTs-channelInRestricted}.

Assume that $(K,E)$ is simplex-embeddable and let $m_x$ be a collection of measurements in $(K,E)$. Let $\rho \in K$, we then have
\begin{equation}
m(\rho) = (m \circ \kappa^* \circ \iota)(\rho) = (m \circ \kappa^*)(\iota(\rho))
\end{equation}
where $\iota: K \to S_\Lambda$ can be seen as a measurement with effects from $E(K)$ and $m \circ \kappa^*: S_\Lambda \to S_O$ can be seen as classical post-processing. Thus the measurements $m_x$ are $E(K)$-compatible \cite[Subsection 5.2]{HeinosaariMiyaderaZiman-compatibility}.
\end{proof}
Another way to prove \eqref{item:incompat-hierarchy-SEtoEK} is to observe that if $(K,E)$ is simplex embeddable, then there also exists a preparation noncontextual model for $(K,E)$ and we get the wanted result using Proposition~\ref{prop:incompat-EKcompat-NCprepHVM}.

\begin{figure}
\centering
\includegraphics[width=0.4\textwidth]{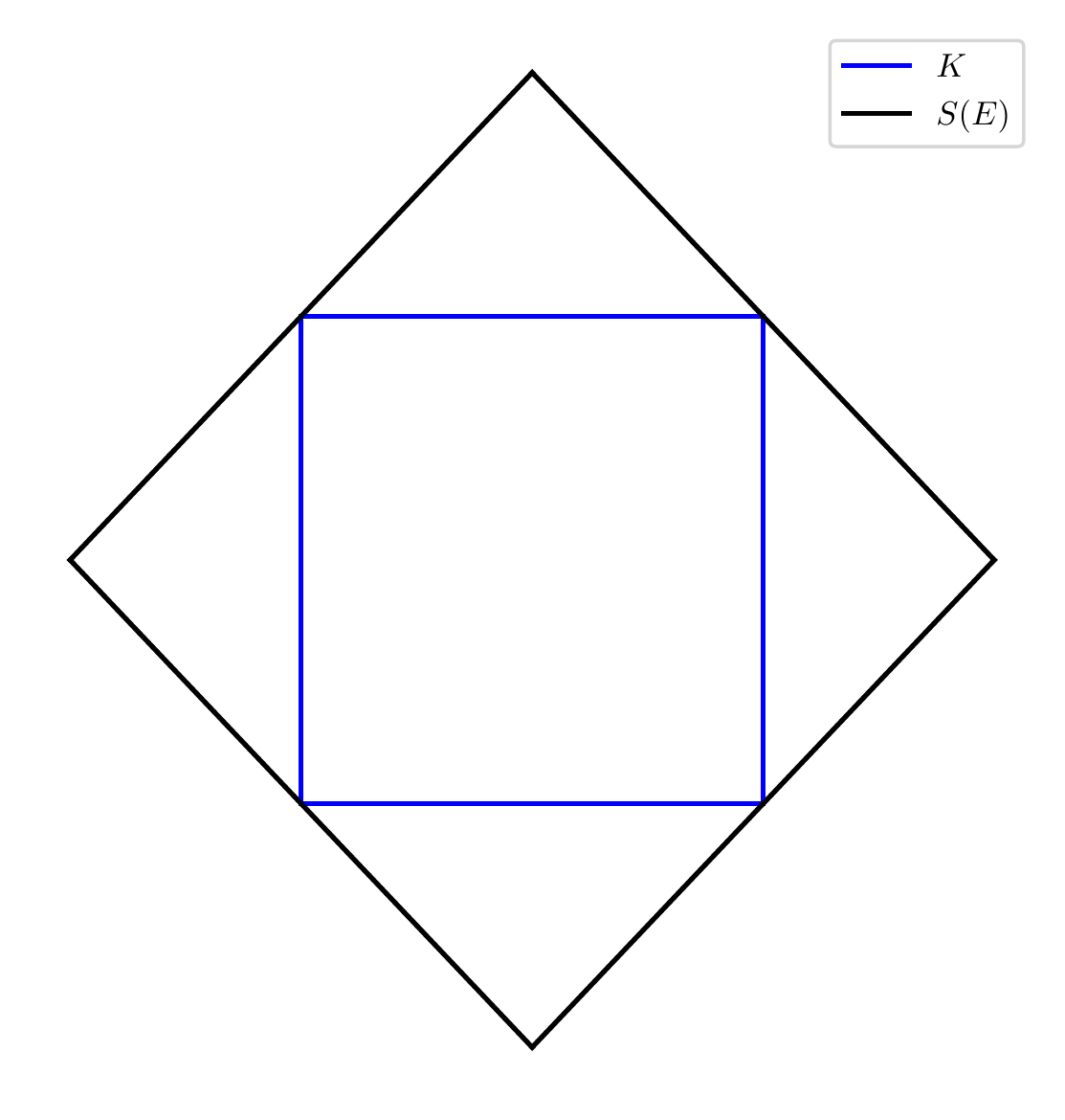}
\caption{Embedding of the square state space $K$ inside the larger square $S(E)$ used in Example~\ref{exm:incompat-squareInSquare}.}
\label{fig:incompat-squareInSquare}
\end{figure}

The following example proves that the inclusion in Theorem~\ref{thm:incompat-hierarchy}~\eqref{item:incompat-hierarchy-EtoSE} is strict.
\begin{example}[Square-in-square] \label{exm:incompat-squareInSquare}
Example of a theory that is simplex-embeddable, but not $E$-compatible is the stabilizer rebit theory theory used in \cite{SchmidSelbyWolfeKunjwalSpekkens-noncontextuality}. Let both $K$ and $S(E)$ be squares, such that $K$ is the convex hull of the midpoints of edges of $S(E)$, see Figure~\ref{fig:incompat-squareInSquare}. It is well-known that there are $E$-incompatible two-outcome measurements \cite{BuschHeinosaariSchultzStevens-compatibility,JencovaPlavala-maxInc}, since $S(E)$ is a square and $E(S(E)) = E$. It was shown in \cite{SchmidSelbyWolfeKunjwalSpekkens-noncontextuality} that $(K,E)$ is simplex-embeddable.
\end{example}

Using the result of Proposition~\ref{prop:incompat-EKcompat-NCprepHVM} and Theorem~\ref{thm:incompat-hierarchy} we get the following corollaries about theories with no restrictions.
\begin{corollary} \label{coro:incompat-noRestrictionSE}
Let $(K, E(K))$ be a theory satisfying the no-restriction hypothesis and assume that $(K, E(K))$ is simplex-embeddable. Then $K$ is a simplex.
\end{corollary}
\begin{proof}
If $(K, E(K))$ is simplex-embeddable, then according to Theorem~\ref{thm:incompat-hierarchy} all measurements in $(K,E(K))$ are $E(K)$-compatible and so $K$ must be a simplex \cite{Plavala-simplex, Kuramochi-simplex}.
\end{proof}

\begin{corollary} \label{coro:incompat-noRestrictionNCprepHVM}
Let $(K, E(K))$ be a theory satisfying the no-restriction hypothesis and assume that there is preparation noncontextual hidden variable model for $(K,E(K))$. Then $K$ is a simplex.
\end{corollary}
\begin{proof}
The result immediately follows from Proposition~\ref{prop:incompat-EKcompat-NCprepHVM}: the existence of preparation noncontextual hidden variable model for $(K, E(K))$ implies that all measurements are $E(K)$-compatible, which implies that $K$ is a simplex \cite{Plavala-simplex, Kuramochi-simplex}.
\end{proof}

\section{Connecting steering and noncontextuality} \label{sec:steering}
In this section we will apply the concepts of $E(K)$-compatibility to the problem of a steerability of a bipartite state $\rho_{AB} \in K_{AB}$. In order to do so, we first provide a short review on steering in GPTs and then we show how to construct the restricted theory $(K,E)$ corresponding to a given steering scenario. For an in-depth review on steering in quantum theory see \cite{UolaCostaNguyenGuhne-steering}.

Let $\rho_{AB} \in K_{AB}$ be a bipartite state shared between Alice and Bob and let $m_x$ be a set of measurements that Alice can choose from, each given by a set of effects $f_{a|x} \in E(K_A)$, such that $\sum_a f_{a|x} = 1_{K_A}$. After measuring $m_x$ and obtaining the outcome $a$, Bob is left with the post-measurement state $\sigma_{a|x}$ given as $\sigma_{a|x} = (f_{a|x} \otimes \id_B)(\rho_{AB})$. Note that if $K_A = \dens(\Ha_A)$, $K_B = \dens(\Ha_B)$ and $K_{AB} = \dens(\Ha_A \otimes \Ha_B)$, then the measurements $m_x$ are given by quantum effects $M_{a|x} \in \effect(\Ha_A)$ and $\sigma_{a|x} = \Tr_A ((M_{a|x} \otimes \I) \rho_{AB})$. In general $\sigma_{a|x} \notin K_B$ since they are subnormalized as $1_K(\sigma_{a|x})$ is the probability of Alice obtaining the outcome $a$ when measuring $x$. The set of states $\{\sigma_{a|x}\}$ is called an assemblage and we say that the assemblage $\{\sigma_{a|x}\}$ is not steerable if there exists a local hidden state model \cite{WisemanDohertyJones-nonlocal}, that is if there is set of states $\{\sigma_\lambda\} \subset K_B$, $\lambda \in \Lambda$, and probability distributions $p(\lambda)$ and $p(a|x, \lambda)$ such that $\sigma_{a|x} = \sum_{\lambda \in \Lambda} p(\lambda) p(a|x, \lambda) \sigma_\lambda$. It is known that a necessary conditions for steering are that the state $\rho_{AB}$ is not separable and that the measurements $m_x$ are incompatible. It is of special interest do decide whether a given state $\rho_{AB} \in K_{AB}$ is steerable by some set of measurements $m_x$, or whether $\rho_{AB}$ is steerable by all possible measurements. We will not restrict the cardinality of the set $m_x$, hence one can take $m_x$ to include all possible measurements.

For the purposes of the following proof we will introduce the notion of full-dimensionality of $\rho_{AB}$.
\begin{definition}
Let $\rho_{AB} \in K_{AB}$. We say that $\rho_{AB}$ is $K_B$-full-dimensional if
\begin{equation}
\linspan(K_B) = \linspan(\{ (f \otimes \id_B)(\rho_{AB}) : f \in E(K_\rho) \})
\end{equation}
where $K_\rho$ is given as in \eqref{eq:steering-Krho}.
\end{definition}
The following is an immediate result.
\begin{lemma} \label{lemma:steering-fullDimSufficient}
Let $\rho_{AB} \in K_{AB}$. $\rho_{AB}$ is $K_B$-full-dimensional if
\begin{equation}
\linspan(K_B) = \linspan( \{ (f \otimes \id_B)(\rho_{AB}) : f \in E(K_A) \} )
\end{equation}
\end{lemma}
\begin{proof}
The result follows from $E(K_A) \subset E(K_\rho)$, which follows from $K_\rho \subset K_A$.
\end{proof}
It is easy to see that in quantum theory if the dimensions of the local Hilbert spaces are the same, then every state with full Schmidt rank is both $K_A$-full-dimensional and $K_B$-full-dimensional.

Let us now build the restricted theory $(K,E)$ corresponding to the steering scenario. Let $\rho_{AB } \in K_{AB}$ and let $g \in E(K_B)$, then $\frac{(\id_A \otimes g)(\rho_{AB})}{(1_{K_A} \otimes g)(\rho_{AB})} \in K_A$ and so we define
\begin{equation} \label{eq:steering-Krho}
K_\rho = \left\lbrace \frac{(\id_A \otimes g)(\rho_{AB})}{(1_{K_A} \otimes g)(\rho_{AB})} : g \in E(K_B) \right\rbrace.
\end{equation}
Let $\{ m_x \}$ be the set of measurements that we want to use for steering. Note that we can without loss of generality assume that $\{ m_x \}$ is closed with respect to post-processing of measurements, as if the set $\{ m_x \}$ does not steer a state $\rho_{AB}$, then we can also construct the local hidden state model for every measurement that is a post-processing of $m_x$. We will define $E_m$ to be the smallest effect algebra containing all of the effects corresponding to $m_x$, i.e.,
\begin{equation} \label{eq:steering-Em}
E_m = \mathrm{effect}( m_x ).
\end{equation}
It follows that if $m_x$ includes all possible measurements, then $E_m = E(K_A)$. The theory $(K_\rho, E_m)$ does not have to be tomographically complete, take for example $\rho_{AB} = \rho_A \otimes \rho_B$, then $K_\rho$ contains only a single state. In some cases one can get rid of the tomographical incompleteness but there are also pathological cases, both of these cases will be addressed later on.

We can now introduce the result that connects noncontextuality and steering.
\begin{theorem} \label{thm:steering-mainResult}
Let $\rho_{AB} \in K_{AB}$ be $K_B$-full-dimensional state and let $(K_\rho, E_m)$ be the restricted theory as given by \eqref{eq:steering-Krho} and \eqref{eq:steering-Em}. Then $\rho_{AB}$ is not steerable by $m_x$ if and only if there exists a preparation noncontextual hidden variable model for $(K_\rho, E_m)$.
\end{theorem}
\begin{proof}
Assume that a preparation noncontextual hidden variable model for $(K_\rho, E_m)$ exists. According to Theorem~\ref{thm:incompat-hierarchy} the measurements given by $f_{a|x}$ are $E(K_\rho)$-compatible, and so there is a measurement given by effects $h_\lambda \in E(K_\rho)$, $\lambda \in \Lambda$, such that $f_{a|x} = \sum_{\lambda \in \Lambda} p(a| x, \lambda)  h_\lambda$. Let $\omega \in K_\rho$, then there is $g \in E(K_B)$ such that $\omega = \frac{(\id_A \otimes g)(\rho_{AB})}{(1_{K_A} \otimes g)(\rho_{AB})}$
and we get
\begin{eqnarray}
f_{a|x}(\omega) = \sum_{\lambda \in \Lambda} p(a| x, \lambda)  h_\lambda(\omega), \\
(f_{a|x} \otimes g)(\rho_{AB}) = \sum_{\lambda \in \Lambda} p(a| x, \lambda)  (h_\lambda \otimes g)(\rho_{AB}),
\end{eqnarray}
for all $g \in E(K_B)$, which implies that
\begin{equation}
(f_{a|x} \otimes \id_B)(\rho_{AB}) = \sum_{\lambda \in \Lambda} p(a| x, \lambda)  (h_\lambda \otimes \id_B)(\rho_{AB}).
\end{equation}
Denoting $(h_\lambda \otimes \id_B)(\rho_{AB}) = p(\lambda) \sigma_\lambda$ we get
\begin{equation}
(f_{a|x} \otimes \id_B)(\rho_{AB}) = \sum_{\lambda \in \Lambda} p(\lambda) p(a| x, \lambda)  \sigma_\lambda
\end{equation}
which is the local hidden state model for $\rho_{AB}$ and $f_{a|x}$. Hence $\rho_{AB}$ is not steerable by $f_{a|x}$.

Now assume that $\rho_{AB}$ is not steerable by $f_{a|x}$ and so there always is a local hidden state model given as
\begin{equation}
(f_{a|x} \otimes \id_B)(\rho_{AB}) = \sum_{\lambda \in \Lambda} p(\lambda) p(a| x, \lambda) \sigma_\lambda,
\end{equation}
where $\sigma_\lambda \in K_B$ for all $\lambda \in \Lambda$. Since $\rho_{AB}$ is $K_B$-full-dimensional it follows that for every $\sigma_\lambda \in K_B$ there is some $h_\lambda \in \linspan(E(K_\rho))$ such that
\begin{equation} \label{eq:steering-mainResult-sigmaLambda}
p(\lambda) \sigma_\lambda = (h_\lambda \otimes \id_B)(\rho_{AB})
\end{equation}
and so we have
\begin{equation}
(f_{a|x} \otimes \id_B)(\rho_{AB}) = \sum_{\lambda \in \Lambda} p(a| x, \lambda)  (h_\lambda \otimes \id_B)(\rho_{AB}).
\end{equation}
Let $g \in E(K_B)$, then we have
\begin{equation}
(f_{a|x} \otimes g)(\rho_{AB}) = \sum_{\lambda \in \Lambda} p(a| x, \lambda)  (h_\lambda \otimes g)(\rho_{AB})
\end{equation}
which implies
\begin{equation}
f_{a|x}(\omega) = \sum_{\lambda \in \Lambda} p(a| x, \lambda)  h_\lambda(\omega)
\end{equation}
for all $\omega \in K_\rho$. Thus we have
\begin{equation}
f_{a|x} = \sum_{\lambda \in \Lambda} p(a| x, \lambda)  h_\lambda.
\end{equation}
This already is close to the wanted result, we only need to prove that $h_\lambda$ are positive functions and that $\sum_{\lambda \in \Lambda} h_\lambda = 1_{K_\rho}$.

So assume that there is some $\omega \in K_\rho$ and $\lambda \in \Lambda$ such that $h_\lambda(\omega) < 0$. Then there is some $g \in E(K_B)$ such that $\omega = \frac{(\id_A \otimes g)(\rho_{AB})}{(1_{K_A} \otimes g)(\rho_{AB})}$ and we get
\begin{equation}
0 > h_\lambda(\omega) = \frac{(h_\lambda \otimes g)(\rho_{AB})}{(1_{K_A} \otimes g)(\rho_{AB})} = \frac{p(\lambda)}{(1_{K_A} \otimes g)(\rho_{AB})} g(\sigma_\lambda).
\end{equation}
Since $\frac{p(\lambda)}{(1_{K_A} \otimes g)(\rho_{AB})} \geq 0$ it follows that $g(\sigma_\lambda) < 0$, but this is a contradiction because $g \in E(K_B)$ and $\sigma_\lambda \in K_B$. To show that $\sum_{\lambda \in \Lambda} h_\lambda = 1_{K_\rho}$ observe that
\begin{eqnarray}
(1_{K_A} \otimes \id_B)(\rho_{AB}) &= \sum_{a} (f_{a|x} \otimes \id_B)(\rho_{AB}) = \sum_{\lambda \in \Lambda} \sum_a p(\lambda) p(a| x, \lambda) \sigma_\lambda \\
&= \sum_{\lambda \in \Lambda} p(\lambda) \sigma_\lambda = \sum_{\lambda \in \Lambda} (h_\lambda \otimes \id_B)(\rho_{AB})
\end{eqnarray}
Let $g \in E(K_B)$, then we have $(1_{K_A} \otimes g)(\rho_{AB}) = \sum_{\lambda \in \Lambda} (h_\lambda \otimes g)(\rho_{AB})$ which implies that for all $\omega \in K_\rho$ we get $1_{K_A}(\omega) = \sum_{\lambda \in \Lambda} h_\lambda (\omega)$ and the result follows.
\end{proof}

Let us first discuss the requirement that $\rho_{AB}$ is $K_B$-full-dimensional. If $\rho_{AB}$ is not $K_B$-full-dimensional, then not all is lost, as in some cases one can restrict $K_B$ to $K_B \cap \linspan(\{ (f \otimes \id_B)(\rho_{AB}) : f \in E(K_\rho) \})$. Denote $J = \linspan(\{ (f \otimes \id_B)(\rho_{AB}) : f \in E(K_\rho) \})$ and assume that there is a channel $\Phi: K_B \to K_B$ such that for all $x \in K_B$ we have $\Phi(x) \in J$ and $\Phi(v) = v$ for all $v \in J$. Assume there is a local hidden state model $(f_{a|x} \otimes \id_B)(\rho_{AB}) = \sum_{\lambda \in \Lambda} p(\lambda) p(a|x, \lambda) \sigma_\lambda$. Since $(f_{a|x} \otimes \id_B)(\rho_{AB}) \in J$ by definition, by applying $\Phi$ to both sides we get $(f_{a|x} \otimes \id_B)(\rho_{AB}) = \sum_{\lambda \in \Lambda} p(\lambda) p(a|x, \lambda) \Phi(\sigma_\lambda)$. Now since $\Phi(\sigma_\lambda) \in K_B \cap J$ we can without loss of generality replace $K_B$ with $K_B \cap J$, simply because from any local hidden state model given by $\sigma_\lambda \in K_B$ we can construct a new local hidden state model given by $\Phi(\sigma_\lambda) \in K_B \cap J$. Note that it is sufficient that there is a channel $\Phi': K_B \to K_B$ such that $\Phi'(v) = v$ for all $v \in J$, then one can use \cite[Lemma 1]{BarnumBarrettLeiferWilce-noBroadcastingPRL} to construct $\Phi: K_B \to K_B$ with the desired properties.

One may now hope to use similar ideas to extend Theorem~\ref{thm:steering-mainResult} to all states $\rho_{AB} \in K_{AB}$, but the following example shows that this is not possible.

\begin{figure}
\centering
\includegraphics[width=0.6\textwidth]{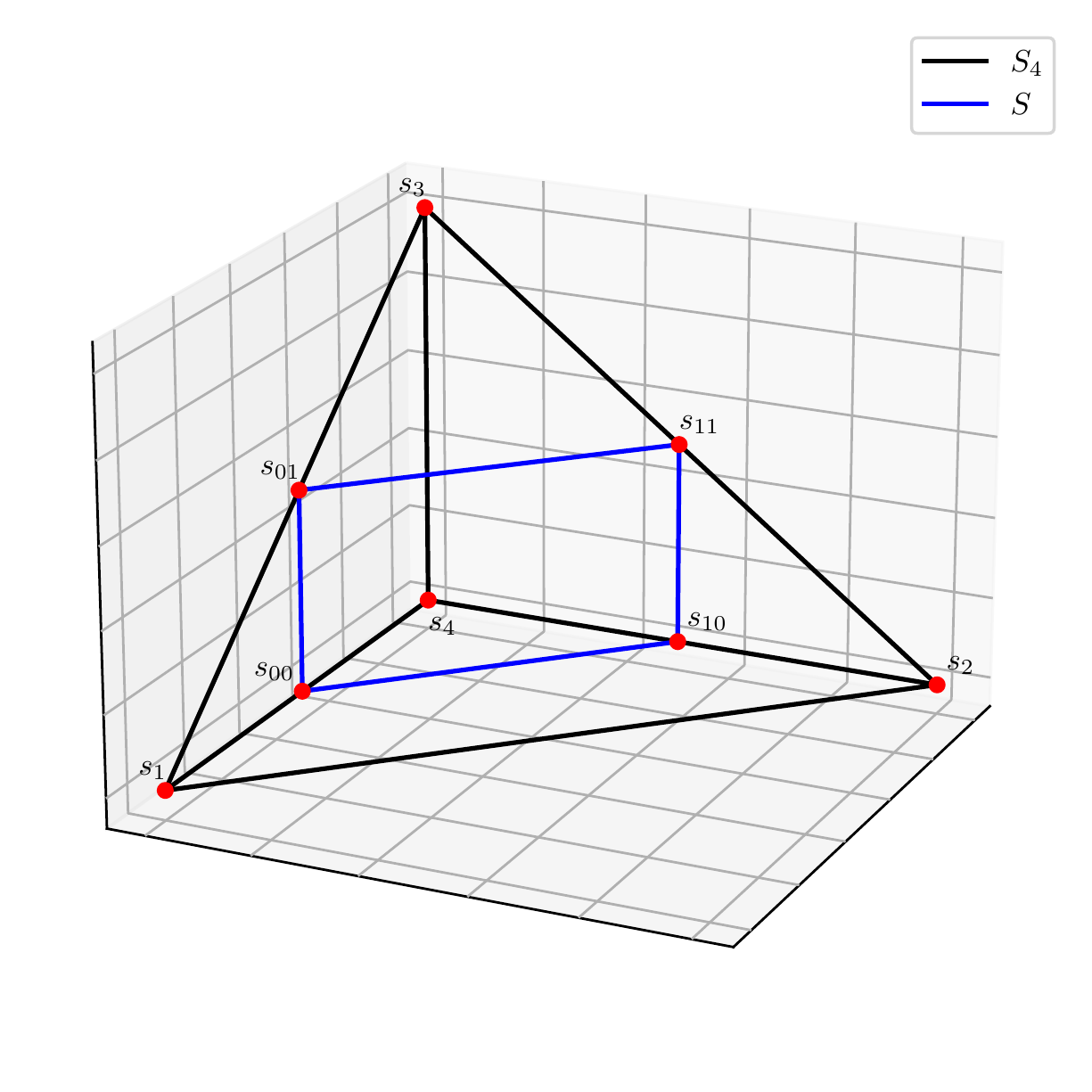}
\caption{Embedding of the square state space $S$ inside the simplex $S_4$ used in Example~\ref{exm:steering-SinS4}.}
\label{fig:steering-SinS4}
\end{figure}

\begin{example}[Sqaure-in-tetrahedron] \label{exm:steering-SinS4}
Let $s_{00}, s_{10}, s_{01}, s_{11}$ be four points such that $\frac{1}{2} (s_{00} + s_{11}) = \frac{1}{2} (s_{10} + s_{01})$, then $S = \conv(\{s_{00}, s_{10}, s_{01}, s_{11}\})$ is, up to an isomorphism, a square. Moreover let $S_4$ be the tetrahedron, i.e., a simplex with four vertices, then there is a subspace $J$ such that $S = S_4 \cap J$, see Figure~\ref{fig:steering-SinS4}. Note that $J \subset \linspan(S_4)$, where $\dim(\linspan(S_4)) = 4$, while Figure~\ref{fig:steering-SinS4} represents only $\aff(S_4)$ as $\dim(\aff(S_4)) = 3$. Let $K_A = S$ and $K_B = S_4$ then we have $S \tmin S_4 = S \tmax S_4$ \cite{AubrunLamiPalazuelosPlavala-cones} and so all bipartite states are separable. This trivially implies that for any measurements $m_x \in \meas(E(K_A)) = \meas(E(S))$ there is a local hidden state model and so no steering occurs. Let $V = \linspan(S)$, then it follows from \cite[Proposition 4.5]{BluhmJencovaNechita-spectrahedra} that $(S \tmin S_4) \cap (V \otimes J) = S \tmax S$. It is known that there are bipartite states $\rho_{AB} \in S \tmax S$ that maximally violate CHSH inequality \cite{PopescuRohrlich-PRbox, Barrett-GPTinformation, BeckmanGottesmanNielsenPreskill-channels, HobanSainz-channels, Crepeau-CHSH, PlavalaZiman-PRbox, JencovaPlavala-PRbox} which necessary implies that there are measurements given by $m_x \in \meas(E(K_A)) = \meas(E(S))$ that steer $\rho_{AB} \in S \tmax S$, because steering is necessary for violations of Bell inequalities \cite{Plavala-channels}. Since $\rho_{AB} \in S \tmax S$ and $S \subset S_4$, we also have $\rho_{AB} \in S \tmin S_4$. It follows that in this case we cannot restrict $S_4$ to $S_4 \cap J$, because for $\rho_{AB} \in S \tmin S_4$ a local hidden variable model always exists, while for the same $\rho_{AB} \in S \tmax S$ it does not.
\end{example}
Note that in the last example we said that $\rho_{AB} \in S \tmin S_4$ is separable while $\rho_{AB} \in S \tmax S$ violates Bell inequality. This is because the measurements that Bob has to measure on $S$ in order to violate the CHSH inequality are not valid measurements on $S_4$. This is simply because $S \subset S_4$ implies $E(S_4) \subset E(S)$ where the inclusion is strict and the missing effects correspond to those measurements that are needed in order to violate the CHSH inequality.

Finally we present an example of steering of the isotropic state. Similar result was also obtained in \cite{HeinosaariKiukasReitznerSchultz-incompatibilityBreakingChannels} using incompatibility breaking channels.
\begin{example}[Steering of the isotropic state] \label{exm:steering-isotropicState}
Let $\Ha$ be a complex Hilbert space, $\dim(\Ha) = 2$ and let $K_A = K_B = \dens(\Ha)$. Let $\rho_{AB}^\gamma = \gamma \ketbra{\phi^+}_{AB} + \frac{1-\gamma}{4} \I_A \otimes \I_B$ be the isotropic state, where $\ket{\phi^+} = \frac{1}{\sqrt{2}}(\ket{00} + \ket{11})$ is the maximally entangled state and $\ket{0}, \ket{1}$ is the computational basis. We will assume that Alice and Bob share the state $\rho_{AB}$ and that Alice has access to all possible measurements. It is known that the isotropic state is not steerable for $\gamma \leq \frac{5}{12}$ \cite{UolaCostaNguyenGuhne-steering}. In order to apply the result of Theorem~\ref{thm:steering-mainResult} we must show that $\rho_{AB}^\gamma$ is $K_B$-full-dimensional. Let $M \in \effect(\Ha)$, which is equivalent to $0 \leq M \leq \I$, then we have
\begin{equation}
\Tr_A ((M \otimes \I) \rho_{AB}^\gamma) = \gamma M^\intercal + (1-\gamma) \Tr(M) \I.
\end{equation}
It is now straightforward to check that for $\gamma > 0$ we have
\begin{equation}
\linspan( \gamma M^\intercal + (1-\gamma) \Tr(M) \I: 0 \leq M \leq I) = \bound(\Ha),
\end{equation}
thus according to Lemma~\ref{lemma:steering-fullDimSufficient} $\rho_{AB}^\gamma$ is $K_B$-full-dimensional for $\gamma > 0$.

Let $0 < \gamma \leq \frac{5}{12}$ then according to Theorem~\ref{thm:steering-mainResult} $(K_{\rho_{AB}^\gamma}, \effect(\Ha))$ is a restricted theory with preparation noncontextual hidden variable model. According to Proposition~\ref{prop:incompat-EKcompat-NCprepHVM} this is equivalent to all measurements in $(K_{\rho_{AB}^\gamma}, \effect(\Ha))$ being $E(K_{\rho_{AB}^\gamma})$-compatible. This is in fact a dual form of the known result that for $0 \leq \gamma \leq \frac{5}{12}$ the depolarizing channel is incompatibility-breaking \cite{HeinosaariKiukasReitznerSchultz-incompatibilityBreakingChannels}, the connection here is that the state $\rho_{AB}^\gamma$ is the Choi matrix of the corresponding depolarizing channel.
\end{example}

\section{Conclusions} \label{sec:conclusions}
We investigated the connection between existence of noncontextual hidden variable models, incompatibility in restricted general probabilistic theories (GPTs) and steering. Our main results are Theorem~\ref{thm:incompat-hierarchy} which connects simplex-embeddability and incompatibility in restricted GPTs and Theorem~\ref{thm:steering-mainResult} which connects steerability of a state with existence of preparation noncontextual hidden variable model for the induced restricted theory. The proof of Theorem~\ref{thm:steering-mainResult} is build on Proposition~\ref{prop:incompat-EKcompat-NCprepHVM} which connects compatibility and existence of preparation noncontextual hidden variable model.

There are several open questions to be addressed in the future, the main is to search for operational equivalents of measurement noncontextual hidden variable models similar to Proposition~\ref{prop:incompat-EKcompat-NCprepHVM} and Theorem~\ref{thm:steering-mainResult}. One can also ask for which GPTs satisfying no-restriction hypothesis measurement noncontextual hidden variable model exists; it is an open question whether this is the case only if the state space of the GPT in question is a simplex, i.e., it is an open question whether equivalent of Corollary~\ref{coro:incompat-noRestrictionSE} and Corollary~\ref{coro:incompat-noRestrictionNCprepHVM} also holds for measurement noncontextual hidden variable models.

\ack
The author is thankful to Chau Nguyen for discussions. This research was supported by the Deutsche Forschungsgemeinschaft (DFG, German Research Foundation, project numbers 447948357 and 440958198), by the Sino-German Center for Research Promotion (Project M-0294), by the ERC (Consolidator Grant 683107/TempoQ), and by the Alexander von Humboldt Foundation.

\bibliographystyle{iopart-num}
\bibliography{citations}

\providecommand{\newblock}{}
\begin{thebibliography}{10}
\expandafter\ifx\csname url\endcsname\relax
  \def\url#1{{\tt #1}}\fi
\expandafter\ifx\csname urlprefix\endcsname\relax\def\urlprefix{URL }\fi
\providecommand{\eprint}[2][]{\url{#2}}

\bibitem{Bell-ineq}
Bell J~S 1964 {\em Physics (College. Park. Md).\/} {\bf 1} 195 -- 200

\bibitem{KochenSpecker-hiddenVariables}
Kochen S and Specker E 1967 {\em Journal of Mathematics and Mechanis\/} {\bf
  17} 59--87 ISSN 0022-2518

\bibitem{BudroniCabelloGuhneKleinmann-contextuality}
Budroni C, Cabello A, G{\"{u}}hne O, Kleinmann M and Larsson J~{\AA} 2021
  (\textit{Preprint} \eprint{2102.13036})

\bibitem{SchmidSelbyWolfeKunjwalSpekkens-noncontextuality}
Schmid D, Selby J~H, Wolfe E, Kunjwal R and Spekkens R~W 2021 {\em PRX
  Quantum\/} {\bf 2} 010331 (\textit{Preprint} \eprint{1911.10386})

\bibitem{Muller-review}
M{\"{u}}ller M 2021 {\em SciPost Physics Lecture Notes\/}  28 ISSN 2590-1990
  (\textit{Preprint} \eprint{2011.01286})

\bibitem{Plavala-review}
Plávala M 2021 General probabilistic theories: An introduction
  (\textit{Preprint} \eprint{2103.07469})

\bibitem{PopescuRohrlich-PRbox}
Popescu S and Rohrlich D 1994 {\em Foundations of Physics\/} {\bf 24} 379--385

\bibitem{Barrett-GPTinformation}
Barrett J 2007 {\em Physical Review A\/} {\bf 75} 032304 (\textit{Preprint}
  \eprint{0508211})

\bibitem{CavalcantiSelbySikoraGalleySainz-Witworld}
Cavalcanti P~J, Selby J~H, Sikora J, Galley T~D and Sainz A~B 2021 {Witworld: A
  generalised probabilistic theory featuring post-quantum steering}
  (\textit{Preprint} \eprint{2102.06581})

\bibitem{Spekkens-contextuality}
Spekkens R~W 2005 {\em Physical Review A\/} {\bf 71} 052108 (\textit{Preprint}
  \eprint{0406166})

\bibitem{Spekkens-toyTheory}
Spekkens R~W 2007 {\em Physical Review A\/} {\bf 75} 032110 (\textit{Preprint}
  \eprint{0401052})

\bibitem{ChiribellaYuan-contextuality}
Chiribella G and Yuan X 2014  (\textit{Preprint} \eprint{1404.3348})

\bibitem{SchmidSpekkensWolfe-nonconIneq}
Schmid D, Spekkens R~W and Wolfe E 2018 {\em Physical Review A\/} {\bf 97}
  062103 ISSN 2469-9926 (\textit{Preprint} \eprint{1710.08434})

\bibitem{SchmidSelbyPuseySpekkens-nonconModels}
Schmid D, Selby J~H, Pusey M~F and Spekkens R~W 2020 {A structure theorem for
  generalized-noncontextual ontological models} (\textit{Preprint}
  \eprint{2005.07161})

\bibitem{SchmidSelbySpekkens-causalInferentialTheories}
Schmid D, Selby J~H and Spekkens R~W 2020 {Unscrambling the omelette of
  causation and inference: The framework of causal-inferential theories}
  (\textit{Preprint} \eprint{2009.03297})

\bibitem{GittonWoods-unitSeparability}
Gitton V and Woods M~P 2020 {Solvable Criterion for the Contextuality of any
  Prepare-and-Measure Scenario} (\textit{Preprint} \eprint{2003.06426})

\bibitem{HeinosaariMiyaderaZiman-compatibility}
Heinosaari T, Miyadera T and Ziman M 2015 {\em Journal of Physics A:
  Mathematical and Theoretical\/} {\bf 49} 123001 (\textit{Preprint}
  \eprint{1511.07548})

\bibitem{GuhneHaapasaloKraftPellonpaaUola-incompatibility}
G{\"{u}}hne O, Haapasalo E, Kraft T, Pellonp{\"{a}}{\"{a}} J~P and Uola R 2021
  {Incompatible measurements in quantum information science} (\textit{Preprint}
  \eprint{2112.06784})

\bibitem{FilippovHeinosaariLeppajarvi-compatibility}
Filippov S~N, Heinosaari T and Lepp{\"{a}}j{\"{a}}rvi L 2017 {\em Physical
  Review A\/} {\bf 95} 032127 (\textit{Preprint} \eprint{1609.08416})

\bibitem{Jencova-incomaptibility}
Jen{\v{c}}ov{\'{a}} A 2018 {\em Physical Review A\/} {\bf 98} 012133
  (\textit{Preprint} \eprint{1705.08008})

\bibitem{BluhmJencovaNechita-incomaptibility}
Bluhm A, Jen{\v{c}}ov{\'{a}} A and Nechita I 2020 {Incompatibility in general
  probabilistic theories, generalized spectrahedra, and tensor norms}
  (\textit{Preprint} \eprint{2011.06497})

\bibitem{UolaMoroderGuhne-steering}
Uola R, Moroder T and G{\"{u}}hne O 2014 {\em Physical Review Letters\/} {\bf
  113} 160403 (\textit{Preprint} \eprint{1407.2224})

\bibitem{QuintinoVertesiBrunner-steering}
Quintino M~T, V{\'{e}}rtesi T and Brunner N 2014 {\em Physical Review
  Letters\/} {\bf 113} 160402 ISSN 0031-9007 (\textit{Preprint}
  \eprint{1406.6976})

\bibitem{BrunnerCavalcantiPironioScaraniWehner-BellNonlocality}
Brunner N, Cavalcanti D, Pironio S, Scarani V and Wehner S 2014 {\em Reviews of
  Modern Physics\/} {\bf 86} 419--478 (\textit{Preprint} \eprint{1303.2849})

\bibitem{TavakoliUola-contextuality}
Tavakoli A and Uola R 2020 {\em Physical Review Research\/} {\bf 2} 013011
  (\textit{Preprint} \eprint{1905.03614})

\bibitem{SelbySchmidWolfeSainzKunjwalSpekkens-contextuality}
Selby J~H, Schmid D, Wolfe E, Sainz A~B, Kunjwal R and Spekkens R~W 2021
  (\textit{Preprint} \eprint{2106.09045})

\bibitem{FilippovGudderHeinosaariLeppajarvi-restrictions}
Filippov S~N, Gudder S, Heinosaari T and Lepp{\"{a}}j{\"{a}}rvi L 2020 {\em
  Foundations of Physics\/} {\bf 50} 850--876 (\textit{Preprint}
  \eprint{1912.08538})

\bibitem{DonadiHossenfelder-superdeterminism}
Donadi S and Hossenfelder S 2020 {A Toy Model for Local and Deterministic
  Wave-function Collapse} (\textit{Preprint} \eprint{2010.01327})

\bibitem{HossenfelderPalmer-superdeterminism}
Hossenfelder S and Palmer T 2020 {\em Frontiers in Physics\/} {\bf 8}
  (\textit{Preprint} \eprint{1912.06462})

\bibitem{Hardy-timeSymmetry}
Hardy L 2021 {Time Symmetry in Operational Theories} (\textit{Preprint}
  \eprint{2104.00071})

\bibitem{CarmeliHeinosaariToigo-incWitness}
Carmeli C, Heinosaari T and Toigo A 2019 {\em Physical Review Letters\/} {\bf
  122} 130402 (\textit{Preprint} \eprint{1812.02985})

\bibitem{Plavala-simplex}
Pl{\'{a}}vala M 2016 {\em Physical Review A\/} {\bf 94} 042108
  (\textit{Preprint} \eprint{1608.05614})

\bibitem{Kuramochi-simplex}
Kuramochi Y 2020 {\em Positivity\/} (\textit{Preprint} \eprint{1912.00563})

\bibitem{BuschHeinosaariSchultzStevens-compatibility}
Busch P, Heinosaari T, Schultz J and Stevens N 2013 {\em Europhysics Letters\/}
  {\bf 103} 10002 (\textit{Preprint} \eprint{1210.4142})

\bibitem{JencovaPlavala-maxInc}
Jen{\v{c}}ov{\'{a}} A and Pl{\'{a}}vala M 2017 {\em Physical Review A\/} {\bf
  96} 022113 (\textit{Preprint} \eprint{1703.09447})

\bibitem{UolaCostaNguyenGuhne-steering}
Uola R, Costa A~C~S, Nguyen H~C and G{\"{u}}hne O 2020 {\em Reviews of Modern
  Physics\/} {\bf 92} 015001 (\textit{Preprint} \eprint{1903.06663})

\bibitem{WisemanDohertyJones-nonlocal}
Wiseman H~M, Jones S~J and Doherty A~C 2007 {\em Physical Review Letters\/}
  {\bf 98} 140402 (\textit{Preprint} \eprint{0612147})

\bibitem{BarnumBarrettLeiferWilce-noBroadcastingPRL}
Barnum H, Barrett J, Leifer M and Wilce A 2007 {\em Physical Review Letters\/}
  {\bf 99} 240501 (\textit{Preprint} \eprint{0707.0620})

\bibitem{AubrunLamiPalazuelosPlavala-cones}
Aubrun G, Lami L, Palazuelos C and Plavala M 2019 Entangleability of cones
  (\textit{Preprint} \eprint{1911.09663})

\bibitem{BluhmJencovaNechita-spectrahedra}
Bluhm A, Jen{\v{c}}ov{\'{a}} A and Nechita I 2020 {Incompatibility in general
  probabilistic theories, generalized spectrahedra, and tensor norms}
  (\textit{Preprint} \eprint{2011.06497})

\bibitem{BeckmanGottesmanNielsenPreskill-channels}
Beckman D, Gottesman D, Nielsen M~A and Preskill J 2001 {\em Phys. Rev. A\/}
  {\bf 64} 052309 ISSN 1050-2947 (\textit{Preprint} \eprint{0102043})

\bibitem{HobanSainz-channels}
Hoban M~J and Sainz A~B 2018 {\em New J. Phys.\/} {\bf 20} 053048
  (\textit{Preprint} \eprint{1708.00750})

\bibitem{Crepeau-CHSH}
Crépeau C, Li J and Yang N 2017 private communication

\bibitem{PlavalaZiman-PRbox}
Pl{\'{a}}vala M and Ziman M 2020 {\em Physics Letters A\/} {\bf 384} 126323
  (\textit{Preprint} \eprint{1708.07425})

\bibitem{JencovaPlavala-PRbox}
Jen{\v{c}}ov{\'{a}} A and Pl{\'{a}}vala M 2020 {\em Physical Review A\/} {\bf
  102} 042208 (\textit{Preprint} \eprint{1907.08933})

\bibitem{Plavala-channels}
Pl{\'{a}}vala M 2017 {\em Physical Review A\/} {\bf 96} 052127
  (\textit{Preprint} \eprint{1707.08650})

\bibitem{HeinosaariKiukasReitznerSchultz-incompatibilityBreakingChannels}
Heinosaari T, Kiukas J, Reitzner D and Schultz J 2015 {\em Journal of Physics
  A: Mathematical and Theoretical\/} {\bf 48} 435301

\end{thebibliography}

\end{document}